\documentclass[runningheads]{llncs}
\spnewtheorem{definition}{Definition}{\bfseries}{\rmfamily}

\newcommand{\old}[1]{{}}

\usepackage{amssymb,latexsym}

\usepackage{amsthm}
\usepackage{adjustbox}
\usepackage{algorithm}
\usepackage{enumerate}
\usepackage{appendix}
\usepackage{amsmath}
\usepackage[numbers]{natbib}
\makeatletter
\renewcommand\bibsection%
{
  \section*{\refname
  \@mkboth{\MakeUppercase{\refname}}{\MakeUppercase{\refname}}}
}
\makeatother

\usepackage{mathtools}
\usepackage{mathrsfs}
\usepackage{dirtytalk}

\usepackage{etoolbox}
\makeatletter
    \pretocmd{\NAT@citexnum}{\@ifnum{\NAT@ctype>\z@}{\let\NAT@hyper@\relax}{}}{}{}
\makeatother

\usepackage{xspace} 

\usepackage{pgf,tikz}
\usepackage{wrapfig}  
\usepackage{cutwin}  

\usepackage{tcolorbox}

\usepackage{hyperref}
\hypersetup{
    colorlinks=true,
    linkcolor=blue,
    filecolor=violet,  
    citecolor=violet,
    urlcolor=cyan,
    pdftitle={Overleaf Example},
    pdfpagemode=FullScreen,
}

\theoremstyle{plain}

\newcommand{\defproblem}[3]{
	\vspace{1mm}
	\noindent\fbox{
		\begin{minipage}{0.96\textwidth}
			\begin{tabular*}{\textwidth}{@{\extracolsep{\fill}}lr} #1 \\ \end{tabular*}
			{\bf{Input:}} #2  \\
			{\bf{Question:}} #3
		\end{minipage}
	}
	\vspace{1mm}
}

\newcommand{\cecp}{\textsc{Colorful Edge Cover}\xspace}
\newcommand{\cvcp}{\textsc{Colorful Vertex Cover}\xspace}
\newcommand{\ckc}{\textsc{Fair $k$-center}\xspace}
\newcommand{\pvc}{\textsc{Partial Vertex Cover}\xspace}

\newcommand{\bmp}{\textsc{Budgeted Matching}\xspace}
\newcommand{\tmp}{\textsc{Tropical Matching}\xspace}

\newcommand{\I}{\mathcal{I}}

\newcommand{\iflong}[1]{}

\newcommand{\mcc}{\mathcal{C}}

\title{On Colorful Vertex and Edge Cover Problems\thanks{A preliminary version of this article appeared in the proceedings of the 47th International Workshop on Graph-Theoretic Concepts in Computer Science (WG) 2021 \cite{DBLP:conf/wg/BandyapadhyayBB21}.}}

\titlerunning{On Colorful Vertex and Edge Cover Problems}

\date{}

\author{Sayan Bandyapadhyay\inst {1}\and
Aritra Banik\inst{2}\and 
Sujoy Bhore\inst{3}}

\institute{Department of Computer Science, Portland State University\\ \email{sayan.bandyapadhyay@uib.no}
\and
School of Computer Sciences, NISER, Bhubaneswar, India\\ \email{aritrabanik@gmail.com}
\and
Department of Computer Science and Engineering, Indian Institute of Technology Bombay, Mumbai, India. \\ \email{sujoy@cse.iitb.ac.in}
}

\authorrunning{S. Bandyapadhyay, A. Banik, S. Bhore}

\begin{document}
\maketitle
\begin{abstract}
In this paper, we study two generalizations of \textsc{Vertex Cover} and \textsc{Edge Cover}, namely \textsc{Colorful Vertex Cover} and \textsc{Colorful Edge Cover}. 
In the \textsc{Colorful Vertex Cover} problem,
given an $n$-vertex edge-colored graph $G$ with colors from $\{1, \ldots, \omega\}$ and coverage requirements $r_1, r_2, \ldots, r_\omega$, the goal is to find a minimum-sized set of vertices that are incident on at least
$r_i$ edges of color $i$, for each $1 \le i \le \omega$, i.e., we need to cover at least
$r_i$ edges of color $i$. \textsc{Colorful Edge Cover} is similar to \textsc{Colorful Vertex Cover}, except here we are given a vertex-colored graph and the goal is to cover at least
$r_i$ vertices of color $i$, for each $1 \le i \le \omega$, by a minimum-sized set of edges. 
These problems have several applications in \emph{fair} covering and hitting of geometric set systems involving points and lines that are divided into multiple groups. Here, \say{fairness} ensures that the coverage (resp. hitting) requirement of every group is fully satisfied.

\qquad We obtain a $(2+\epsilon)$-approximation for the \textsc{Colorful Vertex Cover} problem
in time $n^{O(\omega/\epsilon)}$. Thus, for a constant number of colors, the problem admits a $(2+\epsilon)$-approximation in polynomial time.  
Next, for the \textsc{Colorful Edge Cover} problem, we design an $O(\omega n^3)$ time exact algorithm, via a chain of reductions to a matching problem. For all intermediate problems in this chain of reductions, we design polynomial-time algorithms, which might be of independent interest.  
\end{abstract}

\section{Introduction}

\textsc{Vertex Cover} and \textsc{Edge Cover} are two classical graph problems which have been studied for at least forty years \cite{GareyJ79}. \textsc{Vertex Cover} is known to be \textsf{NP}-complete and admits a 2-approximation \cite{GareyJ79}. On the other hand, \textsc{Edge Cover} can be solved in polynomial time using a connection to \textsc{Maximum Matching} \cite{GareyJ79}. In this paper, we study the following two generalizations of these problems on vertex- or edge-colored graphs. 

\defproblem{\cvcp}{A graph $G$ with $n$ vertices and $m$ edges where every edge is colored by a color from $\{1,2,\ldots,{\omega}\}$, and coverage requirements $r_1,r_2,\ldots,r_\omega$.
}{Find a minimum-sized set of vertices that are incident on at least $r_i$ edges of color $i$, for each $1\le i\le \omega$.}

\defproblem{\cecp}{A graph $G$ with $n$ vertices and $m$ edges where every vertex is colored by a color from $\{1,2,\ldots,{\omega}\}$, and coverage requirements $r_1,r_2,\ldots,r_\omega$.
}{Find a minimum-sized set $E'$ of edges such that at least $r_i$ vertices of color $i$ are incident on the edges of $E'$, for each $1\le i\le \omega$.}

Bera et al.~\cite{bera2014approximation} designed an 
$O(\log \omega)$-approximation for \cvcp. Indeed, they study a more general \say{weighted-version} called \textsc{Partition Vertex Cover}. Moreover, they noted that
an extension of the greedy algorithm of 
Slav\'{i}k~\cite{slavik1997improved} gives an $O(\log(\sum_{t=1}^{\omega}r_t))$ approximation for this problem. On the other hand, it is \textsf{NP}-hard to obtain an approximation guarantee asymptotically better than $O(\log \omega)$ \cite{bera2014approximation}. 
Cohen et al.~\cite{tropical} studied a variant of \cecp where all the requirements are 1 and the solution set of edges $E'$ must form a matching. They gave a polynomial time algorithm for this problem. 

Our motivation to study \cvcp and \cecp partly comes from a series of recent works that study the \ckc problem\footnote{The term \say{fair} stresses on the fact, in an abstract manner, that the resources should be divided evenly among different groups} \cite{DBLP:conf/ipco/AneggAKZ20,DBLP:conf/esa/Bandyapadhyay0P19,DBLP:conf/ipco/JiaSS20}. In \ckc, given a set of $n$ points in a metric space where each point is colored by a color from $\{1,2,\ldots,{\omega}\}$, coverage requirements $r_1,r_2,\ldots,r_\omega$, and an integer $k$, the goal is to find $k$ balls of minimum radius whose union contains at least $r_t$ points of color $t$, for $1 \le t \le \omega$. $O(1)$-approximations are known for this problem when the number of colors $\omega$ is a constant \cite{DBLP:conf/ipco/AneggAKZ20,DBLP:conf/ipco/JiaSS20}. In particular, one can obtain a 4-approximation in $n^{O(\omega)}$ time \cite{DBLP:conf/ipco/AneggAKZ20} and a 3-approximation in $n^{O({\omega}^2)}$ time \cite{DBLP:conf/ipco/JiaSS20}.  

Another motivation is the applications of \cvcp and \cecp to geometric set systems. In the following, we describe two such applications. 

\begin{itemize}
    \item {\bf Covering points by axis-parallel lines.} Suppose we are given a set $\mathcal{L}$ of axis-parallel lines  and a set $P$ of points in the plane, where each point in $P$ is colored by a color from $\{1,\ldots,\omega\}$. Moreover, we are given coverage requirements $r_1,\ldots, r_\omega$. The goal is to find a minimum-sized subset $\mathcal{L}'\subseteq \mathcal{L}$ such that the lines in $\mathcal{L}'$ together contain at least $r_t$ points of color $t$, for each $1 \le t \le \omega$. We note that this problem is a special case of \cvcp : the vertices correspond to the lines and edges correspond to the points -- covering points by lines is same as covering edges by vertices.  
    \vspace{1mm}
    \item {\bf Hitting axis-parallel lines by points.} We are given a set $\mathcal{L}$ of axis-parallel lines and a set $P$ of points in the plane, where each line of $\mathcal{L}$ is colored by a color from $\{1,\ldots,\omega\}$. Also we are given hitting requirements $r_1,\ldots, r_\omega$. The goal is to find a minimum-sized subset $P'\subseteq P$ such that the points in $P'$ intersect at least $r_t$ lines of color $t$, for each $1 \le t \le \omega$. Note that this problem is a special case of \cecp : again the vertices correspond to the lines and edges correspond to the points -- hitting lines by points is same as covering vertices by edges.   
\end{itemize}

\subsection{Our Results}
In this work, we achieve a $(2+\epsilon)$-approximation for \cvcp in time $n^{O(\omega/\epsilon)}$, this means that we obtain an $O(1)$-approximation in polynomial time for constant number of colors, matching the result for \ckc. 
Our algorithm is based on LP rounding and construction of a sparse LP (containing only a few constraints). Sparsity of LPs was also used in the works on \ckc. However, our approach is very different. Indeed, our rounding scheme is less complicated, as in our case each element (an edge) can be covered by only two objects (vertices). This algorithm appears in Section~\ref{sec:2-CVC}. 

We also design an $O(\omega n^3)$-time exact algorithm for the \cecp problem, via a chain of reductions to a matching problem studied by Cohen et al.~\cite{tropical}. The algorithm is described in Section~\ref{sec:3-CEC}.    


\subsection{Related Work} 

Another interesting special case of \cvcp is the \pvc problem, where the value of $\omega$ is equal to 1, i.e., the number of colors is exactly 1. Several polynomial-time 2-approximations are known in this special case via Primal-Dual and LP rounding schemes \cite{bar2001using,bshouty1998massaging,gandhi2004approximation}.  

Inamdar and Varadarajan~\cite{DBLP:journals/corr/abs-1809-06506} studied a generalization of \cvcp, called \textsc{Partition Set Cover (PSC)}. They gave an LP-rounding based $O(\beta + \log \omega)$ approximation, where $\beta$ denotes the approximation guarantee for a related \textsc{Set Cover} instance obtained by rounding the standard LP. 
See also~\cite{bera2014approximation,Har-PeledJ18,DBLP:journals/corr/abs-1809-06506, slavik1997improved} for a comprehensive understanding of this problem. 

Exploiting sparsity of LPs is a popular technique in the design of approximation algorithms. In fact, it has been successfully applied to obtain improved guarantees for several interesting optimization problems, such as $k$-median \cite{li2016approximating,li2017uniform}, $k$-median/means with outliers \cite{krishnaswamy2018constant}, facility location with lower and upper bounds \cite{friggstad2016approximating}, and \ckc \cite{DBLP:conf/ipco/AneggAKZ20,DBLP:conf/esa/Bandyapadhyay0P19,DBLP:conf/ipco/JiaSS20}. 

Related to \cecp, several colored versions of matching problems have been studied in the literature such as \tmp \cite{tropical}, $(\alpha,\beta)$-balanced matching \cite{DBLP:conf/aistats/Chierichetti0LV19}, two-sided matching \cite{freemantwo}, procedurally fair matching \cite{klaus2006procedurally}, etc. Among these the most relevant is \tmp. In fact, we are going to use a known algorithm for \tmp to obtain our result for \cecp.

\section{A $(2+\epsilon)$-approximation for \cvcp}
\label{sec:2-CVC}

First, we describe an LP-rounding based additive approximation for \cvcp, and then show how to convert this to a multiplicative $O(1)$-approximation. Suppose we are given the simple graph $G=(V,E)$ with $V=\{v_1,\ldots,v_n\}$ and $E=\{e_1,\ldots,e_m\}$.  For $1\le t\le \omega$, let $\mcc_t$ denote the color class $t$, i.e., the set of edges of color $t$.  
A vertex $v$ is said to \textit{cover} an edge $e$ if $v$ is an end vertex of $e$. A solution $S$ is said to \textit{cover} an edge $e$ if $S$ contains an end vertex of $e$.  

Next, we describe the natural ILP of \cvcp. For each edge $e_j$, we have a $0/1$ variable $x_j$ that denotes whether $e_j$ is covered in the solution. For each vertex $v_i$, there is a $0/1$ variable $y_i$ that denotes whether $v_i$ is chosen in the solution. There are two main constraints in the ILP other than the domain constraints. The first constraint is the coverage constraint which ensures that from each color class $t$, at least $r_t$ edges are covered. The second constraint is the sanity constraint which ensures that if an edge $e_j$ is covered in the solution, then at least one of its end vertices must be in the solution. The LP relaxation of the ILP is as follows.   

\vspace{2mm}
\begin{tcolorbox}
\begin{align}
\label{LP}
&\text{minimize}&\sum_{v_i \in V} y_i \nonumber \tag{\text{CVC-LP}}
\\&\text{subject to}& \sum_{j:e_j\in \mcc_t} x_j &\ge r_t & & \forall 1\le t\le \omega
\label{constr:color-coverage}
\\& & y_i+y_{i'} & \ge x_j & & \forall e_j=\{v_i,v_{i'}\} \in E \label{constr:sanity}
\\& &0\le x_j, y_i &\le 1 &&\forall e_j \in E, v_i \in V \label{constr:fractional_xy}
\end{align}
\end{tcolorbox}

\vspace{2mm}
We denote any solution to \ref{LP} by $(x,y)$. The cost of $(x,y)$ is defined as, $\text{cost}(x,y)=\sum_{v_i \in V} y_i$. Our rounding algorithm consists of two major steps. 

\paragraph{\textbf{First step.}} In the first step, we compute a fractional optimal solution $(\Bar{x},\Bar{y})$ using any LP solver and modify it to obtain another fractional solution which has a special structure. Let OPT$^{LP}$ denote the cost of $(\Bar{x},\Bar{y})$ and OPT the optimal cost.  

\begin{lemma}\label{lem:separation}
There is a solution $(\Tilde{x},\Tilde{y})$ to $\ref{LP}$ with the following properties: (i) cost$(\Tilde{x},\Tilde{y})\le 2\cdot $\emph{OPT}$^{LP}$. (ii) There is a function $\phi: E \rightarrow V$ such that for each edge $e_j=(v_{j^1},v_{j^2})$, either $\phi(e_j)=v_{j^1}$ or $\phi(e_j)=v_{j^2}$, and $\Tilde{x}_j$ is equal to the $\Tilde{y}$ value of $\phi(e_j)$.   (iii) $(\Tilde{x},\Tilde{y})$ can be obtained in polynomial time. 
\end{lemma}

\begin{proof}
We construct $(\Tilde{x},\Tilde{y})$ by modifying the solution $(\Bar{x},\Bar{y})$. First, we define a function $\phi$ that assigns each edge to a vertex. For each edge $e_j=(v_{j^1},v_{j^2})\in E$, we assign $e_j$ to $v_{j^1}$ or $v_{j^2}$, whichever has the larger $y$-value in $(\Bar{x},\Bar{y})$. If both $y$-values are same, we assign $e_j$ to one of the two arbitrarily. This completes the description of the assignment $\phi$. Next, we construct a new solution based on $\phi$. For each edge $e_j$, we set its new $x$-value to the minimum of 1 and two times the $y$-value of $\phi(e_j)$, i.e., $\Tilde{x}_j=\min\{1,2\Bar{y}_{i'}\}$ where $v_{i'}=\phi(e_j)$. For each vertex $v_i$, we set its new $y$-value to the minimum of 1 and two times of its old $y$-value, i.e., $\Tilde{y}_i=\min\{1,2\Bar{y}_i\}$. Note that for each edge $e_j=(v_{j^1},v_{j^2})$, \[\Tilde{x}_j\ge \min\{1,\Bar{y}_{j^1}+\Bar{y}_{j^2}\}\ge \Bar{x}_j.\] Hence, the new solution $(\Tilde{x},\Tilde{y})$ satisfies the coverage constraints. Also with $v_{i'}=\phi(e_j)$, \[\Tilde{x}_j= \min\{1,2\Bar{y}_{i'}\}=\Tilde{y}_{i'}\le \Tilde{y}_{j^1}+\Tilde{y}_{j^2}.\] Thus, the sanity constraints are also satisfied. Moreover, cost$(\Tilde{x},\Tilde{y})$ is at most two times the cost of $(\Bar{x},\Bar{y})$. Hence, Property (i) is satisfied. Property (ii) is satisfied by construction. Lastly, as the modification takes polynomial time, $(\Tilde{x},\Tilde{y})$ can also be obtained in polynomial time.
\end{proof}

By the above lemma, we obtain a separated LP-solution where each edge gets its fractional coverage $\Tilde{x}_j$ from exactly one of the two end vertices. Based on this separation 
we write a sparse LP (containing only a few constraints) for our instance and use the sparsity of this LP to obtain an integral solution. Next, we describe the details. 

\paragraph{\textbf{Second step.}} Consider the solution $(\Tilde{x},\Tilde{y})$ and the assignment $\phi$ in Lemma \ref{lem:separation}. For each color $1\le t\le \omega$ and for each vertex $v_i \in V$, let $\mcc_{t,i}$ be the set of edges $e_j$ in $\mcc_t$ such that $\phi(e_j)=v_{i}$. Denote the size of $\mcc_{t,i}$ by $m_{t,i}$. Lastly, let $k=\sum_{i:v_i \in V} \Tilde{y}_i$. We define the following LP that does not contain too many constraints. This LP has a variable $z_i$ for each vertex $v_i$, where $1\le i\le n$.  


\vspace{2mm}

\begin{tcolorbox}
\begin{align}
\label{LP1}
&\text{maximize}& \sum_{i=1}^{n} m_{1,i} z_i \nonumber \tag{\text{Sparse-LP}}
\\&\text{subject to}& \sum_{i=1}^{n} m_{t,i} z_i &\ge r_t & & \forall 2\le t\le \omega
\label{constr:sp-color-coverage}
\\& & \sum_{i=1}^{n} z_i \le k \label{constr:sp-cost}
\\& &0\le z_i &\le 1 &&\forall 1\le i\le n \label{constr:fractional_z}
\end{align}
\end{tcolorbox}

\vspace{2mm}

\newpage
\begin{lemma}
There is a solution to \ref{LP1} whose objective function value is at least $r_1$. 
\end{lemma}

\begin{proof}
For each vertex $v_i\in V$, set $z_i=\Tilde{y}_{i}$. Constraint \ref{constr:sp-cost} is trivially satisfied. Now, fix any $1\le t\le \omega$. 
\begin{align*}
    \sum_{i=1}^{n} m_{t,i} z_i  = \sum_{i=1}^{n} m_{t,i} \cdot \Tilde{y}_{i} 
    &  = \sum_{i=1}^{n} \sum_{e_j\in \mcc_{t,i}} \Tilde{y}_{i} & \text{ (as }  m_{t,i}=|\mcc_{t,i}|\text{)}\\
    & = \sum_{i=1}^{n} \sum_{e_j\in \mcc_{t,i}} \Tilde{x}_j & \text{ (from the definitions of } \mcc_{t,i}\text{ and }\phi\text{)}\\
    & =\sum_{j:e_j\in \mcc_t} \Tilde{x}_j & \text{ (as }\mcc_{t,i} \text{ is a partition of } \mcc_t\text{)}\\
    & \ge r_t \qquad & \text{ (by Constraint \ref{constr:color-coverage} of \ref{LP})}
\end{align*}
Hence the lemma follows. 
\end{proof}

Next, we compute a fractional optimal solution $\hat{z}$ to \ref{LP1} using any LP solver. The above lemma implies the value of this solution is at least $r_1$. In the following, we argue about some additional properties of this solution. For that we need the following lemma   
(Lemma 2.1.4 in \cite{lau2011iterative}).

\begin{lemma}[\cite{lau2011iterative}] \label{lemma:ranklemma}
	In any extreme point feasible solution (or equivalently, a basic feasible solution) to a linear program, the number of linearly independent tight constraints is equal to the number of variables.
\end{lemma}

The following lemma is an easy consequence of the above lemma. 

\begin{lemma}\label{lem:almost-integral}
The number of fractional variables in $\hat{z}$ is at most $\omega$. 
\end{lemma}

\begin{proof}
First, note that \ref{LP1} has $2n+\omega$ constraints and $n$ variables. Now, by Lemma \ref{lemma:ranklemma}, the number of linearly independent tight constraints in $\hat{z}$ is $n$. Consider the set $S$ of $2n$ constraints $0 \le z_i\le 1$. As there are only $\omega$ more constraints in the LP, there must be at least $n-\omega$ many linearly independent constraints in $S$ which are tight in $\hat{z}$. Note that the two constraints $z_i\ge 0$ and $z_i\le 1$ cannot be tight simultaneously for any fixed $i$, as it would imply $z_i=0$ and $z_i=1$. Hence, it must be the case that at least $n-\omega$ variables are integral in $\hat{z}$, and the lemma follows.  
\end{proof}

Based on the above lemma we compute an integral solution $z^*$ to \ref{LP1} by rounding the values of the at most $\omega$ fractional variables to 1. Note that this integral solution satisfies all the constraints except Constraint \ref{constr:sp-cost}. But, as we round at most $\omega$ variables, it follows that $\sum_{i=1}^{n} z_i^* \le k+\omega$. Thus, we obtain a set $\Gamma$ of at most $k+\omega$ vertices in $V$ that for each $1\le t\le \omega$, cover at least $r_t$ edges from $\mcc_t$. 
Thus, $\Gamma$ is a feasible solution for \cvcp. By noting that $k=$ cost$((\Tilde{x},\Tilde{y}))\le 2\cdot \text{OPT}^{LP}$, we obtain the following theorem. 

\begin{theorem}\label{thm:additiveapprox}
There is a feasible solution to \cvcp with cost at most $2\cdot$\emph{OPT}$+\omega$ that can be computed in polynomial time. 
\end{theorem}

Next, we show how to convert the above additive approximation to a multiplicative constant approximation, albeit with a time complexity that exponentially depends on $\omega$. 

\begin{theorem}\label{thm:multapprox}
For any $\epsilon > 0$, there is a $(2+\epsilon)$-approximation for \cvcp in $n^{O(\omega/\epsilon)}$ time.
\end{theorem}

\begin{proof}
Fix $\epsilon > 0$. First, we enumerate all the solutions of size $\kappa=1,2,\ldots,\omega/{\epsilon}$ in $n^{O(\omega/\epsilon)}$ time. We stop the first time we obtain a feasible solution. Thus, if we obtain a feasible solution at some step, it must be an optimal solution, and we are done. Otherwise, the optimal cost OPT is more than $\omega/{\epsilon}$. In this case, we use our additive approximation algorithm based on LP rounding. By Theorem \ref{thm:additiveapprox}, we obtain a feasible solution with cost at most $2\cdot $OPT$+\omega < 2\cdot$ OPT$+\epsilon\cdot$ OPT=$(2+\epsilon)\cdot$ OPT. 
\end{proof}

\paragraph{Remark.} The above LP rounding based scheme is much more general in the sense that it also yields an $(f+\epsilon)$-approximation for \textsc{Partition Set Cover} in $n^{O(\omega/\epsilon)}$ time, where each element appears in at most $f$ sets in the input. Here $n$ is the input size. The idea is again similar: assign each element to a unique set having the largest variable value. The $f$ factor comes from the fact that the variable value of each set is scaled up by $f$ factor to obtain the new LP solution where each element is (fractionally) covered by exactly one set.

\section{A polynomial Time Algorithm for \cecp}
\label{sec:3-CEC}

In this section, we study the \cecp problem and obtain a polynomial time exact algorithm. 
In particular, the algorithm runs in $O(\omega n^3)$ time. 

An edge $e$ is said to cover a vertex $v$ if $e$ is incident on $v$. A set of edges $E'$ covers the set of vertices $V'=\{v \mid \exists e\in E' \text{ such that } e \text{ covers } v\}$. 
In the rest of this section, we design the algorithm for \cecp, which is based on an algorithm for the following matching problem. 

\defproblem{\bmp}{A graph $G$ with $n$ vertices and $m$ edges where every vertex is colored by a color from $\{1,2,\ldots,{\omega}\}$, and coverage requirements $r_1,r_2,\ldots, r_\omega$.}{Find a minimum-sized matching which covers at least $r_i$ vertices of color $i$ for each $1\le i\le \omega$.}

We design a polynomial time algorithm for \bmp. But before that, we have the following observation which establishes a connection between \cecp and \bmp.

\begin{lemma}\label{bmp-cecp}
	If \bmp can be solved in time $T(n,\omega)$, then \cecp can be solved in time $T(2n,\omega+1)+O(m+n)$.
\end{lemma}

\begin{proof}
	Suppose we would like to solve \cecp on a given instance $\I$ consisting of a vertex-colored graph $G=(V,E)$ and a set of colors $\{1,2,\ldots,{\omega}\}$. WLOG, there is no isolated vertex in $G$. We construct a new instance $\I'$ of \bmp consisting of a vertex-colored graph $G'=(V',E')$ and a set of colors $\{1,2,\ldots,{\omega},C\}$ as follows. 
	
	For each vertex $v\in V$, we add two vertices, $v$ and an auxiliary vertex $a(v)$, to $V'$. The color of $v$ in $G'$ is same as the color of $v$ in $G$ and the color of $a(v)$ is $C$. $E'$ consists of all the edges in $E$ and the edge $(v,a(v))$ for each $v\in V$. For each color $1\le i\le \omega$, the requirement $r_i$ in $\I'$ remains the same as in $\I$. The requirement corresponding to $C$ is set to 0. Note that $|V'|=2n$ and $|E'|=m+n=O(m)$. 
	
	Next, we show that $\I$ has a solution to \cecp with at most $k$ edges if and only if $\I'$ has a solution to \bmp with at most $k$ edges. First, assume that there is a set of edges $E_1\subseteq E$ of size $k$ which is a solution to \cecp. We construct a matching $M$ for $\I'$ from $E_1$. First, note that if there is a path in $E_1$ consisting of three edges, we can always remove the middle edge from the solution without losing any vertex coverage. Thus, WLOG, we can assume that $E_1$ is a collection of star graphs. Consider any such star $S$. We include any arbitrary edge $(s,v)$ of $S$ in $M$ where $s$ is the central vertex of $S$. For any other edge $(s,u)$ in $S$, we include $(u,a(u))$ in $M$. By construction, $M$ is a matching in $G'$ of size at most $k$. Also, all the requirements are trivially satisfied. 
	
	Now, suppose there is a matching $M$ in $G'$ of size $k$ which is a solution to \bmp. We construct a solution $E_1$ for $\I$ from $M$. For each edge $e\in M\cap E$, include $e$ in $E_1$. For each edge $e\in M\cap (E'\setminus E)$, where $e=(u,a(u))$, include any arbitrary edge $(u,v)$ of $E$ in $E_1$ (that covers $u$). Note that such an edge always exist, as there is no isolated vertex in $G$. Again, by construction, $E_1$ is a feasible solution to $\I$ of size $k$.   
	
	We solve \bmp on $\I'$ to obtain a matching $M$ of the minimum size, say $k$. We return the set of edges $E_1$ as constructed above as the solution to $\I$. We claim that $E_1$ is a solution to \cecp of the minimum size. Suppose it is not. Suppose there is a solution $E_2$ to \cecp of size $k' < k$. Then, by the above discussion, there is a solution to \bmp of size at most $k' < k$. But, this is a contradiction to the assumption that $M$ is a minimum size solution. 
	
	Finally, \bmp can be solved on $\I'$ in $T(2n, \omega+1)$ time, and construction of $G'$ can be done in $O(m+n)$ time. Hence, the lemma follows. 
\end{proof}

In the following, we design an algorithm for \bmp which runs in $O(\omega n^3)$ time. Hence, by the above lemma, we have the following theorem. 

\begin{theorem}\label{cecp-poly}
    \cecp can be solved in $O(\omega n^3)$ time. 
\end{theorem}

To solve \bmp, we show that it can be converted to a problem where each color has unit requirement. Essentially we need the following problem definition due to 
Cohen et al.~\cite{tropical}.

\defproblem{\tmp}{A graph $G$ with $n_1$ vertices and $m_1$ edges where every vertex is colored by a color from $\{1,2,\ldots,{\omega}\}$.
}{Find a maximum-sized matching which covers at least one vertex of color $i$ for each $1\le i\le \omega$.}

We need the following theorem due to 
Cohen et al.~\cite{tropical}.
\begin{theorem}\label{thm:cohen}
	\cite{tropical} \tmp can be solved in $O(n_1 m_1)$ time.  
\end{theorem}

The next lemma establishes the connection between \bmp and \tmp. 

\begin{figure}[!ht]
		\begin{center}
			\includegraphics[width=.8\textwidth]{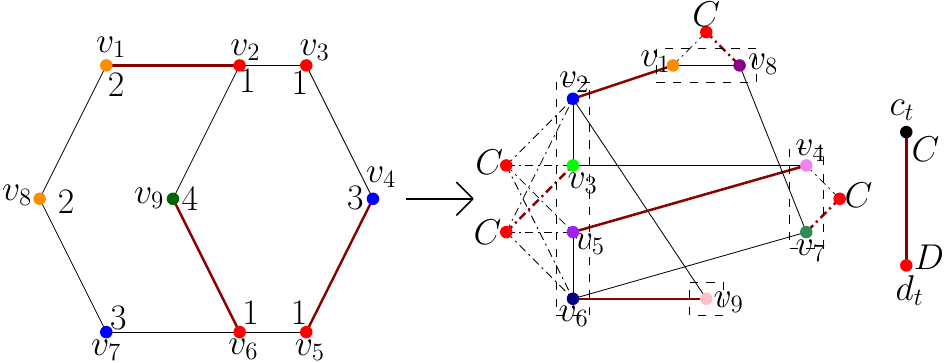}
			\caption{A sample reduction from Budgeted Matching with budget $[2,1,1,1]$ for colors $1,2,3,4$ respectively to Tropical Matching}
			\label{fig:tropical}
		\end{center}
	\end{figure}
	
\begin{lemma}\label{lemma:tropical}
		If \tmp can be solved in $T(n_1,m_1)$ time, \bmp can be solved in $T(\alpha n,\beta\omega n^2)+\gamma\omega n^2$ time for some constants $\alpha,\beta$ and $\gamma$.
\end{lemma}
\begin{proof}
	Suppose we would like to solve \bmp on a given instance $\I$ consisting of a vertex-colored graph $G=(V,E)$ and a set of colors $\{1,2,\ldots,{\omega}\}$. Let $n$ be the number of vertices in $G$. We construct a new instance $\I_t$ of \tmp consisting of a vertex-colored graph $G_t=(V_t,E_t)$ and a set of colors $\{1,2,\ldots,n,C,D\}$ as follows. 
	
	Let $\mcc_x$ be the set of vertices in $G$ of color $x$ and $n_x=|\mcc_x|$ for $1\le x\le \omega$. $V_t$ contains all the vertices in $V$ and for each color $1\le x\le \omega$, a set of $n_x-r_x$ vertices $V^x$. Additionally, $V_t$ contains two more auxiliary vertices $c_t$ and $d_t$. Thus, $V_t=V\cup (\cup_{x=1}^{\omega} V^x) \cup \{c_t,d_t\}$. $E_t$ contains all the edges in $E$ and for each color $1\le x\le \omega$, a set of $(n_x-r_x)\times n_x$ edges $E^x=\{(u,v)\mid u \in V^x \text{ and } v\in \mcc_x\}$. Additionally, the edge $(c_t,d_t)$ is included in $E_t$. Thus, $E_t=E\cup (\cup_{x=1}^{\omega} E^x) \cup \{(c_t,d_t)\}$. Each vertex $u\in V$ in $\I_t$ gets a unique color $j$  for some $1\le j\le n$. Colors of $c_t$ and $d_t$ are $C$ and $D$, respectively. Finally, colors of all vertices in $\cup_{x=1}^{\omega} V^x$ are $C$. See Figure \ref{fig:tropical} for an example construction.  Note that $|V_t|=O(n)$ and $|E_t|=O(\omega n^2)$. 
	
	Next, we show that $\I$ has a solution to \bmp with $k$ edges if and only if $\I_t$ has a solution to \tmp with $n-k+1$ edges. 	First, assume that there is a matching $M$  of size $k$ in $G$ which is a solution to \bmp. We construct a new matching $M_t$ for $G_t$. We include all the edges in $M$ and $(c_t,d_t)$ in $M_t$. For each $1\le x\le \omega$, let $U_x$ be the set of vertices in $\mcc_x$ which are not matched by $M$. We also include a matching between $U_x$ and $V^x$ of size $|U_x|$ in $M_t$. Note that such a matching always exists, as $|U_x|\le n_x-r_x$ by the definition of $M$. Now, we argue that $M_t$ is a valid solution to \tmp of size $(n-k)+1$. First, note that $M_t$ is a matching in $G_t$ which matches all the vertices in $V$. Thus, there is a matched vertex of color $j$ for each $1\le j\le n$. Now, as $(c_t,d_t)$ is also in $M_t$, there are matched vertices of colors $C$ and $D$ as well. Thus, $M_t$ is a feasible solution to \tmp. Note that the size of $\cup_{x=1}^{\omega} U_x$ is exactly $n-2k$, as $k$ edges in $M$ match exactly $2k$ vertices in $V$. Thus, the size of $M_t$ is $k+1+(n-2k)=(n-k)+1$. 
	
	Now, suppose there is a matching $M_t$  of size $(n-k)+1$ in $G_t$ which is a solution to \tmp. We construct a matching $M$ for $G$ starting from $M_t$. Indeed, $M$ is the subset of edges in $M_t$ which are contained in $E$. We argue that $M$ is a feasible solution to \bmp. Note that $M_t$ must match all the vertices in $V$, as each such vertex has a unique color which does not appear in any other vertex. Consider any color $x$ for $1\le x\le \omega$. The edges in $E^x$ can match at most $n_x-r_x$ vertices of $\mcc_x$, as $|V^x|=n_x-r_x$. Thus, there exist at least $r_x$ edges in $M_t\cap E$ which match the vertices in $\mcc_x$ not matched by the edges in $M_t\cap E^x$. It follows that $M$ matches at least $r_x$ vertices of $\mcc_x$ for each $1\le x\le \omega$, and hence it is a feasible solution to \bmp. Next, we show that the size of $M$ is exactly $k$. Let $k_1$ and $k_2$ be the number of edges of $M_t$ which are in $\cup_{x=1}^{\omega} E^x$ and $E$, respectively. Note that $(c_t,d_t)$ must be included in $M_t$, as otherwise there will be no vertex of color $D$ in $M_t$. It follows that $k_1+k_2=n-k$, or $n=k_1+k_2+k$. Now, the $k_1$ edges of $M_t$ in $\cup_{x=1}^{\omega} E^x$ match $k_1$ vertices of $V$, and the $k_2$ edges of $M_t$ in $E$ match exactly $2k_2$ vertices of $V$. As these $k_1+k_2$ edges match all the vertices of $V$, $k_1+2k_2=n=k_1+k_2+k$. It follows that $k_2=k$ making the size of $M$ exactly $k$. 
	
	We solve \tmp on $\I_t$ to obtain a matching $M_t$ of the maximum size, say $s$. We return the matching $M=M_t\cap E$ as the solution to $\I$. We claim that $M$ is a solution to \bmp of the minimum size. Suppose it is not. From the above discussion, we know that the size of $M$ constructed this way, is $n-s+1$. Suppose there is a solution $M'$ to \bmp of size $z < n-s+1$. Then, by the above discussion, there is a solution to \tmp of size $n-z+1 > n+1-n+s-1=s$. But, this is a contradiction to the assumption that $M_t$ is a maximum size solution. 
	
	Finally, \tmp can be solved on $\I_t$ in $T(\alpha n,\beta\omega n^2)$ time for some constants $\alpha,\beta$, and construction of $G_t$ can be done in $\gamma\omega n^2$ time for some constant $\gamma$. Hence, \bmp can be solved on $\I$ in time $T(\alpha n,\beta\omega n^2)+\gamma\omega n^2$. 
 \end{proof}

From Theorem \ref{thm:cohen} and Lemma \ref{lemma:tropical}, we obtain the following theorem.

\begin{theorem}\label{thm:bmp}
	 \bmp can be solved in $O(\omega n^3)$ time.  
\end{theorem}

Since, in Lemma~\ref{bmp-cecp}, we already established the relation between \bmp and \cecp, 
it completes the proof of Theorem~\ref{cecp-poly}.

\bibliographystyle{abbrv}
\bibliography{Arxiv.bib}

\end{document}